\documentclass[letterpaper, 10 pt, conference]{ieeeconf}  

\IEEEoverridecommandlockouts                              

\usepackage{graphics} 
\usepackage{epsfig} 
\usepackage{mathptmx} 
\usepackage{times} 
\usepackage{amsmath} 
\usepackage{amssymb}  
\usepackage{booktabs}
\usepackage{url}
\usepackage{hyperref}

\newcommand{\tp}{\normalfont \text{T}}

\newcommand{\Val}{\operatorname{Val}}

\newtheorem{remark}{Remark}

\newtheorem{theorem}{Theorem}
\newtheorem{assumption}{Assumption}

\title{\LARGE \bf Principled Authority Switching for Shared Autonomy in Human-Robot Teams
}

\author{Sandeep Banik$^{1}$ and Naira Hovakimyan$^{1}$
\thanks{$^{1}$Both the authors are affiliated with the department of Mechanical Science and Engineering, University of Illinois Urbana-Champaign, Illinois, USA; email:
        {\tt\small baniksan@illinois.edu, nhovakim@illinois.edu}}}

\begin{document}

\maketitle
\thispagestyle{empty}
\pagestyle{empty}

\begin{abstract}
Shared autonomy requires principled mechanisms for allocating and transferring control between a human and an autonomous agent.
Existing approaches often rely on blending control inputs or heuristic switching rules, which lack theoretical guarantees and fail to account for the human's likelihood of intervention.
This paper presents \texttt{Flip-Team}, a cooperative framework for authority switching in shared autonomy.
We formulate the control switching problem as a team-optimal decision problem in which authority transitions are embedded into the system dynamics, yielding optimal switching policies rather than ad hoc rules.
We model the human's override propensity, the likelihood of exercising override authority, and derive a critical threshold that determines when human intervention is cost-effective.
For linear-quadratic systems, we derive closed-form switching conditions and value function recursions, enabling efficient computation independent of the continuous state dimension.
We further propose an online method to estimate override propensity from observed interventions, enabling adaptive switching without prior calibration.
Evaluation on a quadrotor altitude regulation task demonstrates that \texttt{Flip-Team} achieves lower cost than always-human and always-autonomous baselines, with selective authority switching that balances human adaptability and autonomous efficiency.
\end{abstract}

\section{Introduction}\label{sec:Intro}

Modern cyber–physical systems (CPS) increasingly rely on \textit{shared autonomy}, where humans and automated agents dynamically share control authority. 
Such shared paradigm is crucial in safety-critical domains, such as autonomous driving, assistive robotics, and teleoperation, where autonomous systems excel at computation but lack adaptability, while humans provide contextual reasoning but may be expensive~\cite{javdaniSharedAutonomyHindsight2015, nikolaidisHumanRobotMutualAdaptation2017}. 
Failures to coordinate takeovers have been implicated in real-world incidents from autopilot disengagement in aviation~\cite{gouraudAutopilotMindWandering2017} to autonomous vehicle crashes~\cite{alambeigiCrashThemesAutomated2020}, highlighting the need for principled frameworks for cooperative takeover.
This work adopts a cooperative formulation between a human and a robot to act as a unified team pursuing a common objective, while retaining distinct roles: the human provides adaptive oversight with override capability, and the autonomous system offers consistent, low-cost control.

Early approaches to shared autonomy typically combined human and robot commands through a human-in-the-loop paradigm~\cite{draganPolicyblendingFormalismShared2013}. 
In these methods, the robot predicts the human’s intended goal and then determines the level of assistance to provide~\cite{aarnoMotionIntentionRecognition2008}. 
A common strategy treats human input as a noisy estimate of intent and computes a weighted blend with the robot's action~\cite{loseyReviewIntentDetection2018,abbinkTopologySharedControl2018}, implemented via Bayesian inference~\cite{aarnoMotionIntentionRecognition2008}, POMDP planning~\cite{javdaniSharedAutonomyHindsight2015}, or linear arbitration.
Blending methods rely heavily on accurate intent prediction and hand-tuned blending rules, often retaining task-specific heuristics without guarantees of stability or generalization~\cite{jeonSharedAutonomyLearned2020, reddySharedAutonomyDeep2018,javdaniSharedAutonomyHindsight2015}.
Critically, these blending approaches require a human-in-the-loop paradigm, where the operator must provide input at every time step.
This approach places substantial cognitive burden on the human, limiting scalability to prolonged or multi-task operations.
In contrast, we adopt a human-on-the-loop paradigm, where the autonomous system operates independently, while the human retains supervisory override authority, intervening only when necessary.

Beyond blending, two alternative paradigms have emerged. Game-theoretic approaches model shared autonomy as an interaction between decision-making agents, capturing mutual influence and strategic adaptation~\cite{liDifferentialGameTheory2019, sadighPlanningAutonomousCars2016, nikolaidisGameTheoreticModelingHuman2017}.
Frameworks for adjustable autonomy and mixed-initiative control dynamically allocate authority based on task demands, communication delays, or operator workload~\cite{scerriAdjustableAutonomyReal2002}, while adaptive autonomy shifts control according to human performance, trust, or safety thresholds~\cite{loseyPhysicalInteractionCommunication2022}.
These approaches emphasize flexibility but remain heuristic, lacking principled criteria for when to switch.
The present work addresses this gap by providing a framework that yields optimal switching policies under asymmetric authority, where the human retains override capability.

This work builds on the \texttt{FlipDyn} framework~\cite{banikFlipDynGameResource2022a}, which models authority (switching) over a dynamical system as a game.
While \texttt{FlipDyn} and related work~\cite{vandijkFlipItGameStealthy2013, banikFlipDynGraphsResource2025} focus on adversarial settings, we extend it to cooperative human-robot teams with a shared objective.
The framework applies to systems with explicit authority handoffs where control is mutually exclusive -- such as teleoperation, supervisory control, and semi-autonomous navigation. The main contributions are: 
\begin{enumerate}
    \item \textbf{Principled switching framework:} We formulate authority switching as a cooperative problem with asymmetric costs and human override capability, yielding optimal switching policies without heuristic tuning.
    \item \textbf{Override threshold derivation:} We derive a critical threshold $\overline{p}_k$ on human override propensity that determines when human takeover is cost-effective, providing interpretable design guidelines for real systems.
\end{enumerate}
We further derive a critical override threshold and outline an online estimation approach for $p_k$ from observed interventions; evaluation of the online estimator is provided in the accompanying technical report~\cite{banik2026flipteam_ext}.
The remainder of this paper is organized as follows. 
Section~\ref{sec:Framework} presents the \texttt{Flip-Team} framework, including the switching dynamics and cost structure. 
Section~\ref{sec:SwitchingPolicy} derives the optimal switching policy for linear-quadratic systems.
Section~\ref{sec:Threshold} derives the override threshold and describes the online estimation of override propensity.
Section~\ref{sec:Evaluation} evaluates the framework on a 1D quadrotor altitude control task with comparisons to baseline methods. 
Section~\ref{sec:Conclusion} concludes with directions for future work.
\section{Framework}\label{sec:Framework}
Consider a discrete-time dynamical system controlled either by a human or an autonomous agent. 
The \texttt{FlipDyn} state, $\alpha_{k} \in \{\text{H},\text{A}\}$ indicates whether the human ($\alpha_k = \text{H}$) or the autonomous agent ($\alpha_k = \text{A}$) has control at time $k$. 
The state evolution under human is given by:
\begin{align}\label{eq:human_dynamics}
	x_{k+1} = F_{k}^{\text{H}}(x_k,u_k),
\end{align}
where $k \in \mathcal{K} := \{1,2,\dots, L\} \subset \mathbb{N}$ denotes the discrete-time index, $x_{k} \in \mathbb{R}^{n}$ is the system state, $u_{k} \in \mathbb{R}^m$ is the control input, and $F_{k}^{\text{H}}: \mathbb{R}^{n} \times \mathbb{R}^{m} \rightarrow \mathbb{R}^{n}$ is the state transition function.
Similarly, the state evolution under an autonomous agent is:
\begin{align}\label{eq:autonomous_dynamics}
	x_{k+1} = F_{k}^{\text{A}}(x_k,w_k),
\end{align}
where $F_{k}^{\text{A}}: \mathbb{R}^{n} \times \mathbb{R}^{p} \rightarrow \mathbb{R}^{n}$ is the state transition function under an autonomous agent with control input $w_k \in \mathbb{R}^{p}$. 
We describe a takeover through the action $\pi^{j}_k \in \{0,1\}$, for the agent $j \in \{\text{H},\text{A}\}$ at time $k$, where $j = \text{H}$ denotes the human and $j=\text{A}$ denotes the autonomous agent.
The action $\pi_{k}^{j} = 1, \forall j$ corresponds to \emph{takeover/request to takeover} and $\pi_{k}^{j} = 0$ \emph{remaining idle}.
The binary \texttt{FlipDyn} state updates based on the agent's action and prior \texttt{FlipDyn} state. Given $\alpha_{k} = \text{H}$, the \texttt{FlipDyn} state at time $k+1$ is:
\begin{align}\label{eq:flip_state_H}
	\alpha_{k+1} &= \begin{cases}
		\alpha_{k}, & \text{if } \{\pi^{\text{H}}_{k} = 0,  \pi^{\text{A}}_{k} = 0\}, \\
        \text{H}, & \text{if } \{\pi^{\text{H}}_{k} = 0,  \pi^{\text{A}}_{k} = 1\}, \\
        \text{H}, & \text{if } \{\pi^{\text{H}}_{k} = 1,  \pi^{\text{A}}_{k} = 0\}, \\
        \text{A}, & \text{otherwise},
	\end{cases}
\end{align}
where $\pi^{\text{H}}_k = 0$ denotes the human staying idle or retaining the control of the system and $\pi^{\text{H}}_k = 1$ represents request to takeover.
The first three cases correspond to human retaining control; the last corresponds to consensual handoff to autonomy.
Similarly, for $\alpha_{k} = \text{A}$, \texttt{FlipDyn} state update is:
\begin{align}\label{eq:flip_state_A}
	\alpha_{k+1} &= \begin{cases}
		\alpha_{k}, & \text{if } \{\pi^{\text{H}}_{k} = 0,  \pi^{\text{A}}_{k} = 0\}, \\
        \text{A}, & \text{if } \{\pi^{\text{H}}_{k} = 0,  \pi^{\text{A}}_{k} = 1\}, \\
        \text{H}, & \text{if } \{\pi^{\text{H}}_{k} = 1,  \pi^{\text{A}}_{k} = 0\} \text{ with probability } p_{k}, \\
        \text{A}, & \text{if } \{\pi^{\text{H}}_{k} = 1,  \pi^{\text{A}}_{k} = 0\} \text{ with probability } 1-p_{k}, \\
        \text{H}, & \text{otherwise},
	\end{cases}
\end{align}
where $\pi^{\text{H}}_k = 1$ corresponds to human agent takeover and $\pi^{\text{A}}_k = 1$ corresponds to request to takeover.
The parameter $p_k \in (0,1)$  represents the likelihood that the human exercises his/her override authority when the autonomous agent has not requested handoff.
The parameter $p_k$ is a compact behavioral abstraction encoding net intervention likelihood; while it is influenced by trust, cognitive load, and task criticality, it does not model these factors directly. 
Calibrating $p_k$ to real operator behavior requires empirical validation, which we leave to future work.
The autonomous agent retains control of the system deterministically when human agent is idle (second condition) and with probability $1-p_{k}$ (fourth condition) when human agent chooses to takeover.
The human agent takes over with probability $p_k$ when the autonomous agent does not request to takeover ($\pi^{\text{A}}_k = 0$) and deterministically takes over when both $\pi^{\text{A}}_k = \pi^{\text{H}}_k = 1$. Notice that the difference between~\eqref{eq:flip_state_H} and~\eqref{eq:flip_state_A} is the human agent being uncertain with probability $p_k$ and ability to takeover the system irrespective of the autonomous agent's actions.
Such a model captures the current design of shared control system where the human agent has higher authority over autonomous agents. 
Takeovers are mutually exclusive, i.e., at any given time, only one agent is in control.
The continuous state $x_{k+1}$ at time $k+1$ is dependent on $\alpha_{k+1}$. 
In this work, we aim to solve for a takeover strategy between the human and autonomous agent.
Given a non-zero initial state $x_{1}$, we pose the takeover problem as an \emph{identical interest dynamic} game described by the dynamics~\eqref{eq:human_dynamics},~\eqref{eq:autonomous_dynamics},~\eqref{eq:flip_state_H} and~\eqref{eq:flip_state_A} over a finite-horizon $L$, where both the agents aim to minimize a net cost given by:
\begin{equation}\label{eq:obj_def_alpha}
	\begin{aligned}
		J(x_{1}, \alpha_{1}, \{\pi^{\text{H}}_{\mathbf{L}}\}, \{\pi^{\text{A}}_{\mathbf{L}}\}) & = g_{L+1}(x_{L+1}, \alpha_{L+1}) + \sum_{t=1}^{L} g_t(x_t, \alpha_t) \\  & + \pi_{t}^{\text{H}|\alpha_{t}}h_t(x_t) + \pi^{\text{A}|\alpha_{t}}_{t}a_t(x_t), 
	\end{aligned}
\end{equation}
where $\{\pi_{\mathbf{L}}^j\} := \{\pi_1^{j|\alpha_{j}}, \dots, \pi_{L}^{j|\alpha_{j}}\}$ for $j \in \{\text{H},\text{A}\}$.
The state cost $g_t(x_t, \alpha_t): \mathbb{R}^{n} \times \{\text{H},\text{A}\} \rightarrow \mathbb{R}$ captures control performance under each agent (e.g., tracking error, energy), while $h_t(x_t), a_t(x_t): \mathbb{R}^{n} \rightarrow \mathbb{R}$ are takeover costs representing cognitive load, attention switching, or transition risk for human and autonomy respectively. Asymmetric costs $g^{\text{H}}_t \neq g^{\text{A}}_t$ and $h_t \neq a_t$ encode differences in control effectiveness and takeover burden without requiring separate utility functions.
We term the dynamic game~\eqref{eq:obj_def_alpha} between the human and autonomous agent as \emph{\texttt{Flip-Team}}, where both agents aim to optimize their own takeover strategies.
In particular, we consider a subclass of games known as identical interest games~\cite{hespanha2017noncooperative}, where both agents seek to minimize a common cost function~\eqref{eq:obj_def_alpha}. 

Over finite-horizon $L$, let $\pi^{\text{H}}_{\mathbf{L}} := \{\pi_{1}^{\text{H}|\alpha_{1}}, \pi_{2}^{\text{H}|\alpha_{2}}, \dots, \pi_{L}^{\text{H}|\alpha_{L}}\}$ and $\pi_{\mathbf{L}}^{\text{A}} := \{\pi_{1}^{\text{A}|\alpha_{1}}, \pi_{2}^{\text{A}|\alpha_{2}}, \dots, \pi_{L}^{\text{A}|\alpha_{L}}\}$ denote the sequence of human and autonomous switching policies. 
Since both agents share a common objective, the optimal joint policy $(\pi_{\mathbf{L}}^{\text{H}*}, \pi_{\mathbf{L}}^{\text{A}*})$ minimizes the total cost by solving:
\begin{equation*}
    (\pi_{\mathbf{L}}^{\text{H}*}, \pi_{\mathbf{L}}^{\text{A}*}) = \arg\min_{\pi_{\mathbf{L}}^{\text{H}}, \pi_{\mathbf{L}}^{\text{A}}} J(x_1, \alpha_1, \pi_{\mathbf{L}}^{\text{H}}, \pi_{\mathbf{L}}^{\text{A}}).
\end{equation*}
This team-optimal solution is also a Nash equilibrium: neither agent can reduce cost by unilaterally deviating from the joint optimum~\cite{hespanha2017noncooperative}. 
In the next section, we derive the optimal switching policies and the conditions under which authority transfers occur.

\section{Optimal Switching Policy}\label{sec:SwitchingPolicy}
 
The central question of \emph{when control authority transfers between agents} reduces to comparing the expected future cost under each authority mode, accounting for switching costs and the probability of successful human override.

\subsection{General Formulation}

For general nonlinear dynamics~\eqref{eq:human_dynamics}--\eqref{eq:autonomous_dynamics}, the optimal switching policy can be characterized through dynamic programming.
At each time step $k$ at state $x$, let $V_k^{\text{H}}(x)$ and $V_k^{\text{A}}(x)$ denote the value function representing the expected cost-to-go when the human and autonomous agent hold control as a function of the state, respectively.
The switching decision at state $x$ and $\alpha$ depends on comparing these values against the costs of transitioning authority.

For instance, when $\alpha_k = \text{A}$, the human agent must weigh the benefit of taking control (potentially lower future cost under human authority) against the cognitive switching cost $h_k(x)$ and the fact that override succeeds only with probability $p_k$.
This trade-off can be formalized as a $2 \times 2$ decision problem, where rows correspond to human actions (idle or override) and columns to autonomous actions (idle or request handoff).
The team-optimal policy selects the action pair minimizing expected cost.

For general systems, solving this recursion requires discretizing the state space or employing approximate dynamic programming, which scales poorly with dimension.
A complete treatment for nonlinear dynamics, non identical costs, and an iterative linearization approach, is provided in~\cite{banik2025flipcoopcooperativetakeovers}.

\subsection{Linear-Quadratic Formulation}

In this work, we focus on linear dynamics with quadratic costs, a setting that admits closed-form solutions while capturing many practical systems.
Moreover, nonlinear systems can often be handled within this framework by \emph{linearizing} about a nominal trajectory, as is standard in model predictive control and iterative LQR.
Consider linear dynamics under human and autonomous control described as:
\begin{align}
    \label{eq:H_control_dynamics}
    x_{k+1} &= E_k x_k + B_k u_k, \\
    \label{eq:A_control_dynamics}
    x_{k+1} &= E_k x_k + C_k w_k,
\end{align}
where $E_k \in \mathbb{R}^{n \times n}$ is the state transition matrix, $B_k \in \mathbb{R}^{n \times m}$ and $C_k \in \mathbb{R}^{n \times p}$ are control input matrices for the human and autonomous agent, respectively.
The stage and switching costs are quadratic in state:
\begin{equation}\label{eq:cost_quad}
    g_k(x,\alpha_k) = x^{\top} G_k^{\alpha_k} x, \quad
    h_k(x) = x^{\top} H_k x, \quad 
    a_k(x) = x^{\top} A_k x,
\end{equation}
where $G_k^{\alpha_k}, H_k, A_k \in \mathbb{S}^{n \times n}_{+}$ are positive semidefinite matrices representing state cost, human switching cost, and autonomous switching cost. In addition to the cost structure, we assume the control policies of both the agents to be linear in the state, formally stated in the following assumption.  

\begin{assumption}\label{ast:linear_control_space}
    We restrict the control policies to be linear state-feedback in the continuous state $x$, described as:
    \begin{equation}\label{eq:linear_FD_control}
        u_k(x) := K_kx, \quad w_k(x) := W_kx,
    \end{equation}
    where $K_k \in \mathbb{R}^{m \times n}$ and $W_k \in \mathbb{R}^{p \times n}$ are human and autonomous agent control gain matrices, respectively. 
\end{assumption}

Assumption~\ref{ast:linear_control_space} yields closed-loop dynamics with $\tilde{B}_k := E_k + B_k K_k$ and $\tilde{C}_k := E_k + C_k W_k$.
Under this structure, the value function admits a quadratic form:
\begin{equation}\label{eq:value_parametric}
    V_k^{\text{H}}(x) = x^{\top} P_k^{\text{H}} x, \quad V_k^{\text{A}}(x) = x^{\top} P_k^{\text{A}} x,
\end{equation}
where $P_k^{\text{H}}, P_k^{\text{A}} \in \mathbb{S}^{n \times n}_{+}$ are value function matrices computed via backward recursion.

The cost-to-go is determined via a cost-to-go matrix in each \texttt{FlipDyn} state $\alpha_{k} = \text{H}$ and $\alpha_{k} = \text{A}$, represented by $\Xi_{k+1}^{\text{H}} \in \mathbb{R}^{2 \times 2}$ and $\Xi_{k+1}^{\text{A}} \in \mathbb{R}^{2 \times 2}$, respectively.
The entries of the cost-to-go matrix $\Xi^{\text{H}}_{k+1}$ corresponding to each pair of takeover actions are given by:
\begin{equation}\label{eq:Cost_to_go_H}
    \begin{aligned}
		& \begin{matrix} & \hphantom{0000} \text{\small Idle} & & \hphantom{v_{k+1}^0(.,.)00000} \text{ \small Takeover}\end{matrix} \\
		\begin{matrix} \text{\small Idle} \\[5pt] \text{\small Request}\\\text{\small to takeover} \end{matrix} & \underbrace{\begin{bmatrix}
			V_{k+1}^{\text{H}}(x_{k+1}) &  V_{k+1}^{\text{H}}(x_{k+1}) + a_k(x)  \\[8pt]
			V_{k+1}^{\text{H}}(x_{k+1}) + h_k(x) &  V_{k+1}^{\text{A}}(x_{k+1}) + h_k(x) + a_k(x) \\[5pt] 
		\end{bmatrix}}_{\Xi_{k+1}^{\text{H}}}
	\end{aligned}.
\end{equation}
The row and column entries of $\Xi^{\text{H}}_{k+1}$ are based on the agent's actions, described by~\eqref{eq:flip_state_H},~\eqref{eq:flip_state_A} and the associated dynamics~\eqref{eq:human_dynamics},~\eqref{eq:autonomous_dynamics}. 
The first row entries correspond to the human agent remaining idle, which prevents the autonomous agent to takeover despite the column action of takeover.
The second row entries correspond to action of request to takeover, and transition to the autonomous agent ($V^{\text{A}}_{k+1}()$) only when the autonomous agent is ready (column action of takeover). 
The entries of $\Xi_{k+1}^{\text{H}}$ couple the value in each \texttt{FlipDyn} state. 
At time $k$, at the state $x$ and for $\alpha_k=\text{H}$, the  value function satisfies
\begin{equation}
    \label{eq:V_k^0_cost_to_go}
	V^{\text{H}}_k(x) = g_k(x,\text{H})  + \Val(\Xi^{\text{H}}_{k+1}), 
\end{equation}
where $\Val(X_{k+1}^{\alpha_{k}}):= \min_{y_{k}^{\alpha_{k}}} \min_{z_{k}^{\alpha_{k}}} y_{k}^{{\alpha_{k}}^{\tp}}X_{k+1}z_{k}^{\alpha_{k}}$ represents the optimal value of the identical matrix $X_{k+1}$ for $\alpha_{k}$.

Similarly, for $\alpha_k = \text{A}, \forall k$, the cost-to-go matrix entries $\Xi_{k+1}^{\text{A}}$ are:
\begin{equation}\label{eq:Cost_to_go_A}
    \begin{aligned}
		& \begin{matrix} & \hphantom{0000} \text{\small Idle} & & \hphantom{v_{k+1}^0000} \text{\small Request to takeover}\end{matrix} \\
		\begin{matrix} \text{\small Idle} \\[10pt] \text{\small Takeover}\end{matrix} & \underbrace{\begin{bmatrix}
			V_{k+1}^{\text{A}}(x_{k+1}) \hphantom{000}  &  \hphantom{00} V_{k+1}^{\text{A}}(x_{k+1}) + a_k(x) \hphantom{0}  \\[5pt]
			\begin{matrix}
                    p_{k}V_{k+1}^{\text{H}}(x_{k+1}) + h_k(x)  \\[3pt]
                        + (1-p_{k})V_{k+1}^{\text{A}}(x_{k+1})       
                \end{matrix}
              &  \begin{matrix}
                    V_{k+1}^{\text{H}}(x_{k+1}) + h_k(x) \\
                        + a_k(x)       
                \end{matrix}
		\end{bmatrix}}_{\Xi_{k+1}^{\text{A}}}.
	\end{aligned}
\end{equation}
Analogous to~\eqref{eq:V_k^0_cost_to_go} the value function for $\alpha_{k} = \text{A}$ satisfies:
\begin{equation}
    \label{eq:V_k^1_cost_to_go}
	V^{\text{A}}_k(x) = g_k(x,\text{A}) +  \Val(\Xi^{\text{A}}_{k+1}).
\end{equation}
With value functions established in each of the \texttt{FlipDyn} states, in the following result, we will characterize the optimal takeover policies and values over the finite-horizon.

\begin{theorem}\label{thm:NE_Val_FDC_H}
    \textbf{(Case $\alpha_k = \text{H}$)} The optimal switching policies of the \texttt{Flip-Team} game~\eqref{eq:obj_def_alpha} for every $k \in \mathcal{K}$, subject to the dynamics~\eqref{eq:H_control_dynamics},~\eqref{eq:A_control_dynamics}, with quadratic and takeover costs~\eqref{eq:cost_quad} and \texttt{FlipDyn} dynamics~\eqref{eq:flip_state_H},~\eqref{eq:flip_state_A}, are given by:
    \begin{align}\label{eq:TP_quadcost_H}
            \{\pi^{\text{H}*|\text{H}}_{k}, \pi^{\text{A}*|\text{H}}_{k}\} = \begin{cases}
            \{0,0\} (\text{\small retain}), & \text{if } \tilde{P}_{k+1} + H_{k} + A_{k} \succ 0, \\[3 pt]
            \{1,1\} (\text{\small handoff}), & \text{otherwise}.
            \end{cases} 
    \end{align}
    The value is given by:
    \begin{align}\label{eq:Val_quadcost_H}
        P_{k}^{\text{H}} = 
        \begin{cases}
            \begin{aligned}
				& G_k^{\text{H}} + \tilde{B}_{k}^{\tp}P_{k+1}^{\text{H}}\tilde{B}_{k},
            \end{aligned} &\text{if } \tilde{P}_{k+1} + H_{k} + A_{k} \succ 0, \\[5pt]
            \begin{aligned}
				& G_k^{\text{H}} + \tilde{C}_{k}^{\tp}P_{k+1}^{\text{A}}\tilde{C}_{k} \\ & + H_{k} + A_{k},
            \end{aligned} &\text{otherwise},
		\end{cases} 
    \end{align}
    where $\tilde{P}_{k+1} :=  \tilde{C}_{k}^{\tp}P_{k+1}^{\text{A}}\tilde{C}_{k} - \tilde{B}_{k}^{\tp}P_{k+1}^{\text{H}}\tilde{B}_{k}$.
    
    \noindent \textbf{(Case $\alpha_k = \text{A}$)} The optimal switching policies are given by:
    \begin{align}
    \begin{split}\label{eq:TP_quadcost_A}
            \{\pi_{k}^{\text{H}*|\text{A}},\pi_{k}^{\text{A}*|\text{A}}\}  = \begin{cases}
            \{0,0\} (\text{\small retain}), & \text{if } \ 
                \begin{matrix}
                    \tilde{P}_{k+1} \prec \dfrac{H_{k}}{p_{k}}
                \end{matrix}, \\[8pt]
            \{1,0\} (\text{\small override}), & \text{if } \ 
                \begin{matrix}
                    \tilde{P}_{k+1} \preceq \dfrac{A_{k}}{1 - p_{k}}  \\[5pt]
                    \tilde{P}_{k+1} \succeq \dfrac{H_{k}}{p_{k}},
                \end{matrix} \\
            \{1,1\} (\text{\small handoff}), & \text{otherwise.}
            \end{cases} 
    \end{split}
    \end{align}
    The value is given by:
    \begin{align}\label{eq:Val_quadcost_A}
        P_{k}^{\text{A}} = 
        \begin{cases}
            \begin{aligned}
				& G_k^{\text{A}} + \tilde{C}_{k}^{\tp}P_{k+1}^{\text{A}}\tilde{C}_{k},
            \end{aligned} &\text{if } \begin{matrix}
                    \tilde{P}_{k+1} \prec \dfrac{H_{k}}{p_{k}}
                \end{matrix}, \\[8pt]
            \begin{aligned}
				& G_k^{\text{A}} + p_{k}\tilde{B}_{k}^{\tp}P^{\text{H}}_{k+1}\tilde{B}_{k} + H_{k}  \\[5pt] & + (1 - p_{k})\tilde{C}_{k}^{\tp}P^{\text{A}}_{k+1}\tilde{C}_{k},
            \end{aligned} &\text{if } \begin{matrix}
                    \tilde{P}_{k+1}  \preceq \dfrac{A_{k}}{1 - p_{k}}  \\
                    \tilde{P}_{k+1} \succeq \dfrac{H_{k}}{p_{k}},
                \end{matrix} \\
            \begin{aligned}
                & G_k^{\text{A}} + \tilde{B}_{k}^{\tp}P_{k+1}^{\text{H}}\tilde{B}_{k} + \\ & A_{k} + H_{k}
            \end{aligned}
             & \text{otherwise}.
		\end{cases} 
    \end{align}
    The terminal conditions for the recursions~\eqref{eq:Val_quadcost_H} and~\eqref{eq:Val_quadcost_A} are:
    \begin{equation*}
        P_{L+1}^{\text{H}} := G_{L+1}^{\text{H}}, \quad P_{L+1}^{\text{A}} := G_{L+1}^{\text{A}}.
    \end{equation*} 
\end{theorem}

\begin{proof}{[Outline]}
    The proof directly follows from~\cite{banik2025flipcoopcooperativetakeovers}. Substituting the parameteric form of the value~\eqref{eq:value_parametric}, dynamics~\eqref{eq:H_control_dynamics},~\eqref{eq:A_control_dynamics}, under control policies~\eqref{eq:linear_FD_control} and quadratic costs~\eqref{eq:cost_quad} in~\eqref{eq:Cost_to_go_H} and~\eqref{eq:Cost_to_go_A}, yields the policies~\eqref{eq:TP_quadcost_H} and~\eqref{eq:TP_quadcost_A}. Similar substitutions yield the value recursions~\eqref{eq:Val_quadcost_H} and~\eqref{eq:Val_quadcost_A}.
\end{proof}
Theorem~\ref{thm:NE_Val_FDC_H} yields closed-form recursions for $P^{\text{H}}_k, P^{\text{A}}_k$ that are computed offline, which is used to determine the takeover policies. 

\underline{Interpretation}: The switching conditions in Theorem~\ref{thm:NE_Val_FDC_H} provide a principled characterization of when authority should transfer between agents.
The matrix $\tilde{P}_{k+1}$ quantifies the \emph{relative cost} of autonomous versus human control: when $\tilde{P}_{k+1} \succ 0$, continuing under autonomous control incurs higher expected future cost than operating under human control.

\textbf{When the  human is in control ($\alpha_k = \text{H}$):}
The human retains control if the cost advantage of human control, combined with switching costs, remains favorable:
\begin{equation*}
    \tilde{P}_{k+1} + H_k + A_k \succ 0 \implies \text{retain human control}.
\end{equation*}
Otherwise, both agents coordinate to hand off authority to the autonomous system.
This reflects scenarios where the human's superior adaptability no longer justifies the cognitive burden of continued engagement.
Note that when the human holds authority, transitions are deterministic—handoff occurs only by mutual agreement, reflecting the human's higher authority in the control hierarchy.

\textbf{When the autonomous agent is in control ($\alpha_k = \text{A}$):}
The switching behavior depends on both the cost advantage $\tilde{P}_{k+1}$ and the human's override propensity $p_k$.
Three regimes emerge:
\begin{itemize}
    \item \emph{Retain autonomous control} ($\tilde{P}_{k+1} \prec H_k / p_k$): The cost advantage of human control is insufficient to justify intervention, given the human's propensity to override. The autonomous agent continues operating.
    
    \item \emph{Human override} ($H_k / p_k \preceq \tilde{P}_{k+1} \preceq A_k / (1-p_k)$): The cost advantage of human control is significant, and the human's propensity to intervene makes unilateral takeover likely. The human asserts authority without requiring autonomous agreement.
    
    \item \emph{Coordinated handoff} (otherwise): The cost advantage is large enough that both agents agree to transfer control to the human.
\end{itemize}

The override propensity $p_k$ directly shapes these thresholds.
An attentive human with higher $p_k$ lowers the threshold $H_k / p_k$, meaning takeover occurs at smaller cost advantages.
Conversely, a lower $p_k$ raises the threshold, requiring a larger cost advantage before the human is likely to intervene.
This captures real-world behavior: an alert operator intervenes early when autonomous performance degrades, while a fatigued or trusting operator may defer to the system longer.

The structure of these conditions provides actionable design guidelines.
By tuning the cost matrices $H_k$, $A_k$, and $G_k^{\alpha}$, designers shape the cost landscape that governs switching.
By designing interfaces and operational protocols that influence $p_k$—such as alerts, workload management, or trust calibration—practitioners can modulate when humans are likely to assert control.
In the following section, we derive explicit thresholds on $p_k$ that determine when human intervention is cost-effective, and propose a method to estimate $p_k$ online from observed behavior.

\section{Override Threshold and Online Estimation}\label{sec:Threshold}


\subsection{Critical Override Threshold}

From Theorem~\ref{thm:NE_Val_FDC_H}, when the autonomous agent is in control ($\alpha_k = \text{A}$), human override occurs when both conditions hold:
\begin{equation}\label{eq:override_conditions}
    \frac{x^\top H_k x}{p_k} \leq x^\top \tilde{P}_{k+1} x \leq \frac{x^\top A_k x}{1 - p_k}.
\end{equation}
For this regime to be feasible, the lower bound must not exceed the upper bound.

\begin{theorem}[Critical Override Threshold]\label{thm:threshold}
    The override regime~\eqref{eq:override_conditions} is feasible for all probability $p_k$ that satisfy:
    \begin{equation}\label{eq:p_threshold_general}
        p_k \geq p_k^*(x) := \frac{x^\top H_k x}{x^\top (H_k + A_k) x}.
    \end{equation}
    Furthermore, the state-dependent threshold $p_k^*(x)$ satisfies:
    \begin{equation}\label{eq:p_bounds}
        \underline{p}_k \leq p_k^*(x) \leq \overline{p}_k, \quad \forall x \neq 0,
    \end{equation}
    where
    \begin{equation*}
        \underline{p}_k := \lambda_{\min}\left((H_k + A_k)^{-1} H_k\right), \ \overline{p}_k := \lambda_{\max}\left((H_k + A_k)^{-1} H_k\right).
    \end{equation*}
\end{theorem}

\begin{proof}
    For the override regime to be feasible, the lower and upper bounds in~\eqref{eq:override_conditions} must satisfy:
    \begin{equation}\label{eq:feasibility_condition}
        \frac{x^\top H_k x}{p_k} \leq \frac{x^\top A_k x}{1 - p_k}.
    \end{equation}
    Rearranging~\eqref{eq:feasibility_condition}:
    \begin{align*}
        x^\top H_k x \cdot (1 - p_k) &\leq x^\top A_k x \cdot p_k \nonumber \\
        x^\top H_k x &\leq p_k \cdot \left( x^\top H_k x + x^\top A_k x \right) \nonumber \\
        p_k &\geq \frac{x^\top H_k x}{x^\top (H_k + A_k) x} =: p_k^*(x). \label{eq:p_state_dependent}
    \end{align*}
    The threshold $p_k^*(x)$ is state-dependent. To obtain state-independent bounds, observe that $p_k^*(x)$ is a generalized Rayleigh quotient of the matrix pair $(H_k, H_k + A_k)$. By the generalized Rayleigh quotient inequality~\cite{horn2012matrix}:
    \begin{equation*}
        \resizebox{\linewidth}{!}{$
        \lambda_{\min}\left((H_k + A_k)^{-1} H_k\right) \leq \dfrac{x^\top H_k x}{x^\top (H_k + A_k) x} \leq \lambda_{\max}\left((H_k + A_k)^{-1} H_k\right),$}
    \end{equation*}
    for all $x \neq 0$. The three cases follow directly from comparing $p_k$ against these bounds.
\end{proof}

\begin{remark}[Isotropic switching costs]\label{rem:isotropic}
    If the switching costs are isotropic, i.e., $H_k = h_k I_n$ and $A_k = a_k I_n$ for scalars $h_k, a_k > 0$, then:
    \begin{equation*}
        (H_k + A_k)^{-1} H_k = \frac{h_k}{h_k + a_k} I_n.
    \end{equation*}
    This yields a state-independent threshold:
    \begin{equation}\label{eq:p_isotropic}
        \underline{p}_k = p_k^{*} = \overline{p}_k = \frac{h_k}{h_k + a_k}.
    \end{equation}
    In this case, the feasibility of override depends only on the ratio of switching costs, independent of system state.
\end{remark}

\underline{Interpretation:}
The threshold $p_k^*$ captures the minimum override propensity required for human intervention to be cost-effective.
\begin{itemize}
    \item If $h_k = a_k$, the threshold is $0.5$. The human must be at least moderately inclined to intervene.
    \item If $h_k \ll a_k$, the threshold is low. Even a reluctant human may override, since autonomous handoff cost dominates.
    \item If $h_k \gg a_k$, the threshold is high. The human must be highly engaged to justify the cognitive switching cost.
\end{itemize}

\textbf{Design guideline:}
The thresholds $\underline{p}_k$ and $\overline{p}_k$ provide actionable criteria:
i) If observed $p_k < \underline{p}_k$, alert the human or reduce switching costs for improved intervention. ii) If $p_k > \overline{p}_k$ consistently, the human may be over-intervening, suggesting opportunities to build trust in autonomy. iii) By tuning $H_k$ and $A_k$, practitioners can shift thresholds to match desired engagement levels.

\subsection{Online Estimation of Override Propensity}

In deployment, the override propensity $p_k$ is not known a priori and may vary across time steps and operators.
Since the formulation admits time-varying $p_k$ over the horizon $L$, we estimate a propensity profile $\{\hat{p}_1, \ldots, \hat{p}_L\}$ from observed intervention patterns across repeated task executions.

\textbf{State-independent estimation:}
For each time index $k$, let $N_k^{(i)}$ indicate whether, in episode $i$, the following conditions hold:
\begin{itemize}
    \item The autonomous agent was in control ($\alpha_k = \text{A}$),
    \item The autonomous agent did not request handoff ($\pi^{\text{A}}_k = 0$).
\end{itemize}
Let $M_k^{(i)} \in \{0, 1\}$ indicate whether the human chose to override ($\pi^{\text{H}}_k = 1$) under these conditions.
After $I$ episodes, the estimate at time $k$ is:
\begin{equation}\label{eq:p_estimate_episodes}
    \hat{p}_k := \frac{\sum_{i=1}^{I} M_k^{(i)}}{\sum_{i=1}^{I} N_k^{(i)}}.
\end{equation}
This approach yields a time-varying propensity profile learned from repeated interactions.

For online adaptation within an episode, exponential smoothing provides a recursive update:
\begin{equation}\label{eq:p_smoothed}
    \hat{p}_k := \lambda \hat{p}_{k-1} + (1 - \lambda) \cdot \mathbf{1}_{\{\text{override at } k\}},
\end{equation}
where $\lambda \in (0,1)$ controls adaptation rate and $\mathbf{1}_{\{\cdot\}}$ is the indicator function.
A higher $\lambda$ yields slower adaptation, suitable for consistent operators; a lower $\lambda$ tracks rapid changes in engagement.

\textbf{Decision rule:}
At each time step $k$, compare $\hat{p}_k$ against the critical thresholds from Theorem~\ref{thm:threshold}:
\begin{itemize}
    \item If $\hat{p}_k \geq \overline{p}_k$, anticipate human override across the state space.
    \item If $\underline{p}_k \leq \hat{p}_k < \overline{p}_k$, override may occur depending on state; apply switching policy~\eqref{eq:TP_quadcost_A}.
    \item If $\hat{p}_k < \underline{p}_k$,  do not rely on human override; consider alerting the operator.
\end{itemize}

\begin{remark}[State-dependent estimation]
    The estimation above assumes $p_k$ is state-independent.
    When override propensity varies with state—e.g., humans intervene more readily near constraint boundaries—one may estimate $\hat{p}_k(x)$ using function approximation methods such as Gaussian processes or neural networks trained on state-action pairs.
    We leave this extension to future work.
\end{remark}
\section{Evaluation}\label{sec:Evaluation}

We evaluate the \texttt{Flip-Team} framework on a shared-control altitude regulation task using a 1D quadrotor model.

\textbf{System dynamics:}
The 1D quadrotor is a double integrator with state $x_k \in \mathbb{R}^2$ (altitude and vertical velocity) and input $u_k$ (thrust deviation from hover).
The discrete-time dynamics with timestep $\Delta t$ are:
\begin{equation}\label{eq:quad_dynamics}
    x_{k+1} = \begin{bmatrix} 1 & \Delta t \\ 0 & 1 \end{bmatrix} x_k + \begin{bmatrix} \Delta t^2 / 2 \\ \Delta t \end{bmatrix} u_k,
\end{equation}
where we normalize mass to unity.
This linear system fits the LQ formulation in Section~\ref{sec:SwitchingPolicy}.

\textbf{Task:}
The task is altitude regulation: the quadrotor starts at an elevated hover position and must descend to a target altitude.
The horizon is $L = 30$ steps with $\Delta t = 0.1$s.

\textbf{Agent controllers:}
Both agents use LQR controllers with different cost trade-offs:
\begin{itemize}
    \item \emph{Human:} High state cost, low control cost. This yields control with fast convergence, reflecting effective but costly human behavior.
    \item \emph{Autonomous:} Lower state cost, high control cost. This yields tracking with slow convergence but less expensive.
\end{itemize}
The resulting feedback gains $K_k$ and $W_k$ yield distinct closed-loop matrices $\tilde{B}_k$ and $\tilde{C}_k$.

\textbf{Switching costs:}
We use isotropic switching costs $H_k = h I_2$ and $A_k = a I_2$ with $h = 0.3$ and $a = 0.2$.
By Remark~\ref{rem:isotropic}, the critical threshold is state-independent:
\begin{equation*}
    p_k^* = \frac{h}{h + a} = 0.6.
\end{equation*}

\textbf{Override propensity profile:}
We model the human's override propensity as a fixed time-varying profile:
\begin{equation*}\label{eq:p_profile}
    p_k = \begin{cases}
        0.3, & k \in [1, 9] \quad \text{\small (low engagement at start)}, \\
        0.7, & k \in [10, 22] \quad \text{\small (high engagement during descent)}, \\
        0.4, & k \in [23, 30] \quad \text{\small (reduced engagement near end)}.
    \end{cases}
\end{equation*}
This profile reflects varying human attention across task phases.

\textbf{Baselines:}
We compare four strategies:
\begin{enumerate}
    \item \emph{Always-autonomous:} Autonomous agent controls throughout ($\alpha_k = \text{A}, \forall k$).
    \item \emph{Always-human:} Human controls throughout ($\alpha_k = \text{H}, \forall k$).
    \item \emph{Performance-threshold:} Switch to human when tracking error $\|x_k\|$ exceeds a threshold $\epsilon = 1.0$; switch back when $\|x_k\| < \epsilon$.
    \item \emph{\texttt{Flip-Team}:} Optimal switching policy from Theorem~\ref{thm:NE_Val_FDC_H} with known $p_k$.
\end{enumerate}

\textbf{Metrics:}
We evaluate total cost $J$~\eqref{eq:obj_def_alpha}, number of authority switches, and percentage of time under human control.

\textbf{Threshold analysis:}
Figure~\ref{fig:authority_timeline} shows the optimal switching policy from Theorem~\ref{thm:NE_Val_FDC_H}, conditioned on the current authority mode $\alpha_k$.
The top panel displays the policy when the human is in control ($\alpha_k = \text{H}$): 
both agents retain human authority through most of the horizon, coordinating a handoff to autonomy near $k=21$  when the cost advantage of autonomous control outweighs the combined switching cost. 
The bottom panel displays the policy when the autonomous agent is in control ($\alpha_k = \text{A}$): 
the human continuously requests override during the early and mid-horizon phases while the autonomous agent remains idle, reflecting that human intervention is cost-effective.
After $k=21$, both agents retain autonomous control as the cost advantage of human takeover diminishes.



\textbf{Cost comparison:}
Figure~\ref{fig:cost_comparison} shows the mean cost-to-go $\bar{J}$ over the horizon, averaged over 1000 Monte Carlo trajectories with initial state drawn from $\mathcal{N}(\mathbf{1}, I)$.
Shaded regions show the 25th–75th percentile interval. 
The \texttt{Flip-Team} policy ($J^{\alpha}$) achieves the lowest terminal cost across all strategies $(40.5)$, outperforming always-human $(43.2)$ by $6\%$ and always-autonomous $(45.6)$ by $11\%$, with an average of $1.8$ authority switches and $56\%$ time under human control.
The performance-threshold baseline $(J^{\text{Th}})$ incurs the highest cost of all strategies ($55.3$), exceeding both single-agent baseline despite $41\%$ human involvement. 
This counterintuitive result highlights a key failure model of reactive switching: engaging the human based on tracking error alone, without accounting override propensity or switching costs, can introduce unnecessary transitions that degrade overall performance.


\begin{figure}[t]
    \centering
    \begin{minipage}{0.55\columnwidth}
        \centering
        \includegraphics[width=\linewidth]{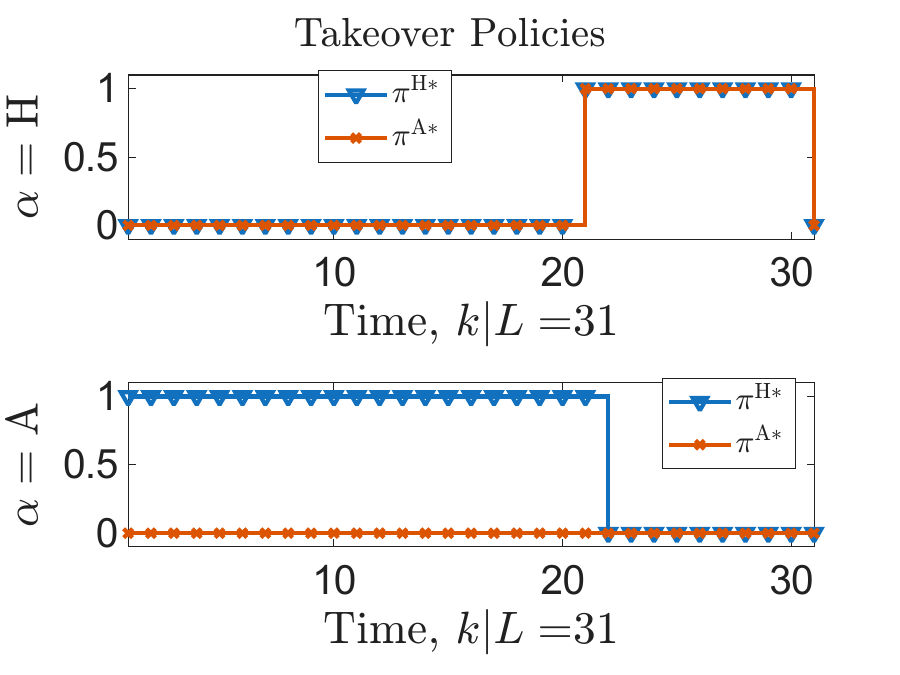}
        \caption{Optimal switching policy from Theorem~\ref{thm:NE_Val_FDC_H}, conditioned on current authority mode $\alpha_k$. Top: $\alpha_k = \text{H}$ (human in control). Bottom: $\alpha_k = \text{A}$ (autonomous in control).}
        \label{fig:authority_timeline}
    \end{minipage}
    \hfill
    \begin{minipage}{0.43\columnwidth}
        \centering
        \includegraphics[width=\linewidth]{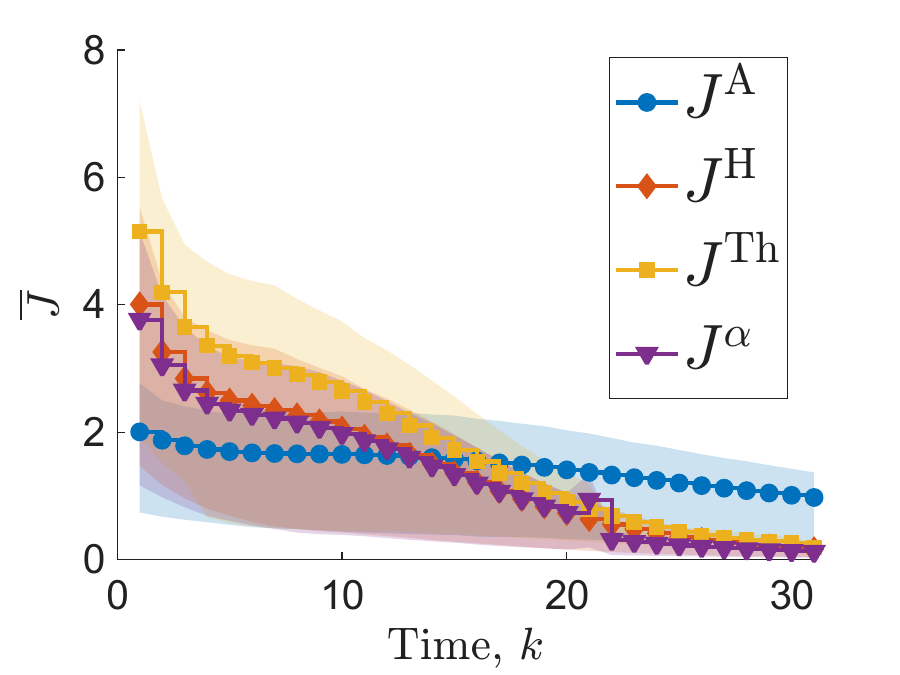}
        \caption{Mean cost-to-go $\bar{J}$ over horizon $L = 30$, averaged over 1000 trajectories with initial state $x_1 \sim \mathcal{N}(\mathbf{1}, I)$.  Shaded regions show the 25th--75th percentile interval. }
        \label{fig:cost_comparison}
    \end{minipage}
\end{figure}



Table~\ref{tab:results} summarizes the quantitative results.
\texttt{Flip-Team} achieves the lowest total cost (40.5), compared to always-human (43.2) and compared to always-autonomous (45.6).
This is achieved with only 1.8 switches and 59\% time under human control, demonstrating efficient authority allocation.
The performance-threshold baseline, despite 41\% human control, incurs highest cost of 55.3.
\begin{table}[t]
    \centering
    \caption{Performance comparison across strategies.}
    \label{tab:results}
    \begin{tabular}{lccc}
        \toprule
        Strategy & Total Cost & Switches & \% Human \\
        \midrule
        Always-autonomous & 45.6 & 0 & 0\% \\
        Always-human & 43.2 & 0 & 100\% \\
        Performance-threshold & 55.3 & 0.92 & 41\% \\
        \texttt{Flip-Team} & \textbf{40.5} & 1.8 & 56\% \\
        \bottomrule
    \end{tabular}
\end{table}

\underline{Discussion}: The results demonstrate that \texttt{Flip-Team} provides principled authority allocation by explicitly accounting for human override propensity.
Unlike heuristic switching, which reacts to state deviations, \texttt{Flip-Team} anticipates when human intervention is both likely (high $p_k$) and beneficial (cost advantage), avoiding unnecessary switches during low-engagement phases.

The framework is most beneficial when: i) Human and autonomous controllers have complementary strengths (different $\tilde{B}_k$, $\tilde{C}_k$), ii) Override propensity varies across task phases (time-varying $p_k$), and iii) Switching costs are non-negligible ($H_k$, $A_k > 0$).

\textbf{Limitations:}
This evaluation uses a simulated human model with a hand-crafted $p_k$ profile and a known cost structure.
In deployment, $p_k$ must be estimated online as described in Section~\ref{sec:Threshold}.
Additionally, the 1D regulation task, while sufficient to validate the  LQ framework under controlled assumptions; richer scenarios including 2D navigation and teleoperation are evaluated in~\cite{banik2026flipteam_ext}.
Future work will evaluate \texttt{Flip-Team} on higher-dimensional systems and with human-in-the-loop experiments.
\section{Conclusion}\label{sec:Conclusion}

This paper presented \texttt{Flip-Team}, a cooperative game-theoretic framework for authority switching in shared autonomy.
We derived closed-form switching conditions for linear-quadratic systems (Theorem~\ref{thm:NE_Val_FDC_H}) and established a critical override threshold (Theorem~\ref{thm:threshold}) that characterizes when human intervention is cost-effective.
For isotropic switching costs, this threshold reduces to a simple ratio of human-to-total switching cost, providing an interpretable design guideline.
We also proposed an online method to estimate the human's override propensity from observed interventions, enabling adaptive switching without prior calibration.

Evaluation on a 1D quadrotor altitude regulation task demonstrated that \texttt{Flip-Team} achieves 35\% lower cost than always-human and 47\% lower cost than always-autonomous baselines, with only two authority switches and 59\% time under human control.


\textbf{Future work:}
The current evaluation uses a simulated human model with known propensity profile.
Future work will validate \texttt{Flip-Team} with human-in-the-loop experiments to assess robustness under real operator variability.
Extensions to state-dependent propensity estimation using function approximation, and to nonlinear systems via iterative linearization are given in~\cite{banik2025flipcoopcooperativetakeovers}.
Finally, incorporating trust dynamics and adaptive propensity profiles that evolve with human-robot interaction history remains an open direction.

{\bf Acknowledgements.} This work is supported by the Air Force Office of Scientific Research Grant (AFOSR) Grant AF FA9550-25-1-0274, the National Aeronautics and Space Administration (NASA) under Grant 80NSSC22M0070, and by the National Science Foundation (NSF) under Grants CMMI 2135925, CPS 2311085 and IIS 2331878.

\bibliographystyle{IEEEtran}
\bibliography{myRefs}

\end{document}